\documentclass[10pt, journal]{IEEEtran}
\normalsize
\usepackage{amsthm,amssymb,amsmath,tikz,graphics,cite,float,graphicx,epstopdf,epsfig,verbatim,url,bbm,mathtools}
\usepackage{siunitx}
\usetikzlibrary{automata}
\usetikzlibrary{shapes,arrows}
\newtheorem{lem}{Lemma}
\usepackage[utf8]{inputenc}
\newtheorem{thm}{Theorem}

\newtheorem{rem}{Remark}
\newtheorem{corol}{Corollary}

\newtheorem{mydef}{Definition}

\usepackage{capt-of}
\usepackage[noend]{algpseudocode}
\usepackage{algorithmicx}
\usepackage[ruled]{algorithm}
\allowdisplaybreaks

\tikzset{
  treenode/.style = {align=center, inner sep=0pt, text centered,
    font=\sffamily},
  arn_r/.style = {treenode, circle, black, draw=black, 
    text width=1.5em, very thick},
}
 
\begin{document}
\title{Secure and Private Cloud Storage Systems with Random Linear Fountain Codes}
\author{Mohsen~Karimzadeh~Kiskani$^{\dag}$ and
         Hamid~R.~Sadjadpour$^{\dag}$
\thanks{M. K. Kiskani$^{\dag}$ and H. R. Sadjadpour$^{\dag}$ 
are with the Department of Electrical Engineering, University of California, Santa Cruz. Email: 
\{mohsen, hamid\}@soe.ucsc.edu}}
\maketitle \thispagestyle{empty}

\begin{abstract}
 An information theoretic approach to security and privacy    called Secure And Private Information Retrieval (SAPIR) is introduced. SAPIR is applied to distributed data storage systems. In this approach, random combinations of all contents are stored across the network. Our coding approach is based on Random Linear Fountain (RLF) codes.
To retrieve a content, a group of servers collaborate with each other to form a Reconstruction Group (RG). SAPIR achieves asymptotic perfect secrecy if at least one of the servers within an RG is not compromised. Further, a Private Information Retrieval (PIR) scheme based on random queries is proposed. The PIR approach
ensures the users  privately download their desired contents without the servers knowing about the requested contents indices. The proposed scheme is adaptive and can provide privacy against a significant number of colluding servers. 
\end{abstract}


\IEEEpeerreviewmaketitle

\section{Introduction}

Cloud networks have become a popular platform for data storage during the past decade. Cloud systems have been used in different applications  such as  healthcare \cite{hosseini2016cloud}. Security of the stored data has always been a major concern for many cloud service providers. Many cloud service providers use encryption algorithms to encrypt the data on their servers. Dropbox, for instance,  is using Advanced Encryption Standard (AES) to store the contents on its servers\footnote{\url{https://www.dropbox.com/en/help/27}}. Since the encryption algorithms are {\em computationally secure}, an adversary may be able to break them with time. For instance, Data Encryption Standard (DES)
which was once the official Federal Information Processing Standard (FIPS) in US is not considered secure anymore. An interesting problem in highly sensitive cloud services would then be to design {\em information theoretic secure} solutions which are immune to attackers in time. 

To achieve perfect information theoretic secrecy using Shannon cipher system \cite{shannon1949communication}, the number of keys should be equal to the number of messages. Therefore, to retrieve the contents from the cloud using an information theoretically secure approach in which the contents are directly encoded with a different key, each user needs to store a huge number of keys which is not practical. In this paper, we propose to use the storage capability of the  trusted servers to generate the keys by using the contents themselves and achieve asymptotic perfect secrecy. Our proposed technique is based on Random Linear Fountain (RLF) codes \cite{mackay2005fountain}. RLF codes have been shown \cite{DBLP:conf/infocom/KiskaniS16, kiskani2017secure,DBLP:journals/twc/KiskaniS17} to be very useful in distributed storage systems.

On the other hand, in many distributed storage applications like Peer-to-Peer (P2P) distributed storage systems or distributed storage systems in which some of the servers are under the control of an oppressive government, a user wants to download a content in a way that the servers cannot determine which content is requested by the user. This is widely known as Private Information Retrieval (PIR) problem. 

Our next contribution in this paper is a novel technique to address the PIR problem in distributed storage systems. Users use random queries to request data from the servers. These random queries are designed in a way that they can be used to retrieve any desired content while preventing any malicious agent with the knowledge of up to half of the random queries to gain information about the requested content. This is an important feature of the proposed technique that provides privacy in the presence of many colluding servers. Such a feature has not been presented in prior information theoretic PIR  approaches \cite{tajeddine2016private} for coded storage systems. {The proposed} Secure And Private Information Retrieval (SAPIR) scheme provides both security and privacy for information retrieval.


The rest of the paper is organized as follows. Section \ref{sec_related} is dedicated to the related work on PIR and security in distributed storage systems.  The assumptions and problem formulation are described in section \ref{sec_problem}. We study the security and PIR aspects of SAPIR in sections \ref{sec_security} and \ref{sec_pir}, respectively.
The simulation results are provided in  section \ref{sec_sim} and the paper is concluded in section \ref{sec_conc}. 

\section{Related Works}
\label{sec_related}

In this paper, we use {\em Random Linear Fountain (RLF) codes} \cite{mackay2005fountain} to encode the contents within the servers in the network. Significant capacity improvement can be achieved in wireless ad hoc and cellular networks \cite{DBLP:conf/infocom/KiskaniS16,kiskani2017secure,DBLP:journals/twc/KiskaniS17, Saeed.Vahidian:IET:2014, hatamnia2017network, kiskani2015application,kiskani2013social, kiskani2016effect,kiskani2017asecure, kiskani2017multihop, kiskani2015opportunistic, mousaei2017optimizing} using RLF codes. The application of fountain codes in distributed storage systems was previously studied in  \cite{DBLP:conf/icassp/DimakisPR06}. Similar coding techniques has been previously used in references like \cite{yazdani2017oncognitive, sani2016distributed, sani2015quantizer, hajizadeh2012asymmetric, kiskani2010delay, kiskani2011novel, kiskani2016location, kiskani2015recent} to provide quality of service in wireless systems.

The capacity of wireless ad hoc cached networks was studied in \cite{kiskani2017secure} and it was shown that RLF codes can achieve perfect secrecy asymptotically without considering the PIR problem. In the current paper, RLF codes are used to simultaneously achieve  security and privacy in distributed cloud storage systems.

While MDS codes \cite{DBLP:journals/tit/DimakisGWWR10,DBLP:journals/pieee/DimakisRWS11} show good repair capability, these codes are not particularly designed to provide security. Authors in \cite{DBLP:conf/isit/DikaliotisDH10} have studied the security of distributed storage systems with MDS codes and \cite{DBLP:journals/icl/KumarRA16} has proposed a construction for repairable and secure fountain codes. {Reference \cite{DBLP:journals/icl/KumarRA16} achieves security by concatenating Gabidulin codes with Repairable Fountain Codes (RFC). Their specific design allows to use Locally Repairable Fountain Codes for secure repair of the lost data.} Unlike \cite{DBLP:journals/icl/KumarRA16} which has focused on the {security of the repair links} using concatenated codes, the current paper presents  simultaneous security and privacy {of the data storage nodes} by only using RLF codes. References \cite{Vahidian:WCL:2015} and \cite{saeedtvt2017} have studied the problem of security in the presence of overhearing interference in cooperative communications. Further \cite{naghshin2015capacity, naghshin2017coverage} studied the same problem on multi-tier networks. 

The authors in \cite{cheng2016numerical} have numerically studied the wiretap network with a simple topology in which there is a relaying node between the transmitter and the receiver. In the current paper, we considered the general network with a cloud infrastructure in which the servers are cooperating to reconstruct the contents. 

The idea of PIR was originally introduced in \cite{DBLP:conf/focs/ChorGKS95} for uncoded databases. Recently, there has been a renewed interest in studying PIR for storage systems utilizing different coding techniques. Reference \cite{DBLP:conf/isit/ShahRR14} was among the first references to study the problem of PIR for coded storage systems. They proved that with only one extra bit, PIR can be achieved. However, the solution in \cite{DBLP:conf/isit/ShahRR14} requires that the number of servers grows with the data record size. Reference \cite{DBLP:conf/isit/ChanHY15} assumed that the number of servers is fixed and established the trade-off between storage and retrieval costs and demonstrated the fundamental limits on the cost of PIR for coded storage systems. The authors in \cite{tajeddine2016private}  studied the problem of PIR for MDS coded storage systems and introduced a scheme to achieve PIR in MDS coded databases but the security aspect was not addressed in that paper.  
They have also assumed that the databases are able to store all the contents which may not be a realistic assumption. Unlike prior work \cite{DBLP:conf/isit/ShahRR14,DBLP:conf/isit/ChanHY15,tajeddine2016private} which have only studied PIR for coded databases, we are interested in achieving simultaneous security and PIR. Further, as far as we know, this is the first work to study the problem of PIR for a fountain coded-based distributed storage system. The proposed PIR scheme is easily scalable to the cases when up to half of the servers are colluding to obtain information about the content or content index which makes this technique very robust against large number of colluding servers. 

\section{Problem Formulation}
\label{sec_problem}
The network is composed of $n$ servers each capable of storing $h$ contents. These servers are denoted by $\mathcal{N}_1,\mathcal{N}_2,\dots,\mathcal{N}_n$. A total number of $m$ contents exist within the network and each content has $M$ bits, i.e., $f_1, f_2, \dots, f_m$.

\subsection{RLF Coding-Based Storage}
The contents are randomly encoded and stored on the servers during the data preloading phase. The encoded file in the $j^{th}$ storage location of the $i^{th}$ server for any $i=1,2,\dots,n$ and $j=1,2,\dots,h$ will have the form
 \begin{align}
  c_j^i = \sum_{k=0}^m v_k^{i,j}  f_k = \mathbf{f} \mathbf{v}_j^i,  
 \end{align}
 where\footnote{Throughout the paper, the vectors are denoted in bold characters.} $\mathbf{f} = [f_1 ~ f_2 ~ \dots f_m]$ denotes the $1 \times m$ vector of all contents and $\mathbf{v}_j^i$ denotes an $m \times 1$ random encoding vector of  $0$'s and $1$'s. Each content $f_i$ belongs to the Galois Field $\mathbb{F}_{2^{M}}$, i.e. $\mathbf{f} \in \mathbb{F}_{2^M}^m$. Throughout the paper, unless otherwise stated we assume that all the vector and matrix operations are  in $\mathbb{F}_2$.  The encoded files stored in server $\mathcal{N}_i$ are $\mathbf{c}_i = [c_1^i~ c_2^i~ \dots c_h^i]$ where $\mathbf{c}_i \in \mathbb{F}_{2^M}^h$.
Note that $\mathbf{c}_i = \mathbf{f} \mathbf{V}_i$ where $\mathbf{V}_i$ is the $m \times h$ random encoding matrix for server $\mathcal{N}_i$.

In RLF all random vectors $\mathbf{v}_j^i$ are chosen independently and uniformly from $\mathbb{F}_2^m$ which results in a random uniform choice of the encoding matrix $\mathbf{V}_i$ where each element can be either $0$ or $1$ with equal probability. Such an encoding matrix may not necessarily be full rank and may contain linearly dependent rows. This will result in redundant use of storage and may jeopardize the security by revealing more information. Hence, we propose a {\em full rank encoding} scheme based on RLF codes in which randomly created encoding vectors $\mathbf{v}_j^i$ are discarded if they already exist in the span of the previously selected random encoding vectors. In other words, for each server we select $h$ linearly independent vectors to construct a full rank matrix $\mathbf{V}_i$ of size $m \times h$ for $i=1,2,\dots,n$. 
 
The encoding can be performed in a decentralized way. This means that each server can fill up its storage space independently of all the other servers during the data preloading phase. It can be shown \cite{DBLP:journals/twc/KiskaniS17} that the average minimum number of encoded files required to decode any desired content is very close to the optimal value of $m$.

\subsection{Reconstruction Groups (RG)}
 After the data preloading phase, users can reconstruct their desired contents during content delivery phase. A desired file $f_r$ can be written as $f_r =  \mathbf{f} \mathbf{e}_r$, where $\mathbf{e}_r$ is an all zero vector of size $m \times 1$ except in the $r^{th}$ location is equal to 1. To retrieve $f_r$, the user needs to access enough encoded files on the network servers in order to construct $\mathbf{e}_r$ via ${\mathbf v}_j^i$'s.

Since codes are constructed in $\mathbb{F}_2^m$, users need $m$ linearly independent encoding vectors to retrieve any of the $m$ contents.
We assume that servers are divided into many different {\em RGs }. Servers within each RG collaborate with each other to retrieve any requested content. Therefore, the number of encoded files within a single RG should be at least equal to $m$. The RGs are represented by $\mathcal{J}_1,\mathcal{J}_2,\dots,\mathcal{J}_u$ and the number of servers within their corresponding RGs by $J_1,J_2,\dots,J_u$ where, $\sum_{i=1}^u J_i = n$. It is shown in \cite{DBLP:journals/twc/KiskaniS17} that the average minimum number of encoded files within each RG to retrieve all the contents is only slightly larger than $m$. Therefore, for each RG $\mathcal{J}_i$ where $1 \le i \le u$, the minimum value of $J_i$ is only slightly larger than $\frac{m}{h}$. Notice that if $J_i$ is smaller than $\frac{m}{h}$, then the servers will not be able to form a full rank matrix to retrieve all desired contents. In the case that storage systems store uncoded contents, we need exactly $m$ cache locations for storing files which is very close to our RLF technique and demonstrates that our RLF-coding based approach efficiently utilizes storage space. For large values of $h$, i.e. $h \ge m$, each server can become an RG by itself.  

\subsection{Content Retrieval}
 Each RG $\mathcal{J}_k$  stores $J_k h \ge m$ randomly encoded files. The matrices $\mathbf{V}_i$ of the $J_k$ servers in the RG form a full rank matrix $\mathbf{V} = [\mathbf{V}_1 ~\mathbf{V}_2~ \dots ~\mathbf{V}_{J_k}]_{m \times J_k h}$. Therefore, any content with index $r$ can be retrieved from the servers  by solving the linear equation $\mathbf{V} \mathbf{y}_r = \mathbf{e}_r$ in $\mathbb{F}_2$. Since this matrix is full rank, one possible solution can be given as 
 \begin{align}
 \mathbf{y}_r = \mathbf{V}^T \left( \mathbf{V} \mathbf{V}^T \right)^{-1} \mathbf{e}_r.
 \label{eqs_simple_response}
 \end{align}
 To solve $\mathbf{V} \mathbf{y}_r = \mathbf{e}_r$, servers within the RG should send their corresponding encoding matrices $\mathbf{V}_i$ to one of the RG servers called $\mathcal{N}_s$ that generates 
  $\mathbf{V}$ and computes $\mathbf{y}_r$ from the above equation\footnote{Notice that the servers of an RG only need to send this information to $\mathcal{N}_s$ once. This could be done even right after the data preloading phase.}. If $\mathbf{y}_r = [\mathbf{y}_r^1~\mathbf{y}_r^2~\dots~\mathbf{y}_r^{J_k}]^T$ is such a solution, where $\mathbf{y}_r^i$ is a $h \times 1$ local decoding vector for server $\mathcal{N}_i$, then server $\mathcal{N}_s$ sends $\mathbf{y}_r^i$ to server $\mathcal{N}_i$ and $\mathcal{N}_i$ then transmits $\mathbf{f} \mathbf{V}_i \mathbf{y}_r^i$ to the requesting user. All of the server responses are then aggregated by the user to retrieve $f_r$ as
 \begin{align}
  f_r = \mathbf{f} \mathbf{e}_r = \mathbf{f} \mathbf{V} \mathbf{y}_r= \sum_{i=1}^{J_k} \mathbf{f} \mathbf{V}_i \mathbf{y}_r^i. 
  \label{eqs_f_r_0}
 \end{align}
 However, this solution reveals the identity of the downloaded content to all the servers of the RG. This simple solution cannot be used for PIR but  we will show in section \ref{sec_security} that perfect secrecy can be achieved with this solution. A solution to preserve the privacy of the users is presented in section \ref{sec_pir}.
 
 \section{Security}
 \label{sec_security}
 This section is dedicated to the study of security of our approach. 
 If an adversary is able to wiretap all of the communication links between the RG servers and the user, it can perfectly retrieve $f_r$ using equation \eqref{eqs_f_r_0}.  We prove that perfect communication secrecy can be achieved when the adversary can wiretap all communication links between servers and user except one. 
 We will prove this for the case when the user directly sends the request $\mathbf{e}_r$ to the servers and the servers respond accordingly. Under this scenario, the adversary knows the requested content index but still unable to reduce its equivocation about the requested content.
 
Consider RG $\mathcal{J}_k$ and without loss of generality, assume that an adversary can wiretap all of the links between servers $\mathcal{N}_1,\mathcal{N}_2,\dots,\mathcal{N}_{J_k-1}$ and the user. Further assume that the user wants to directly download the content $f_r$ from these servers by sending the query $\mathbf{e}_r$ to all these servers. Such a scenario is much more vulnerable to adversarial attacks compared to a scenario in which the requested base vectors are expanded in terms of random queries in order to guarantee privacy. When the query $\mathbf{e}_r$ is received by all the servers,they will collectively solve the linear equation $\mathbf{V} \mathbf{y}_r = \mathbf{e}_r$ to find the decoding vector $\mathbf{y}_r$. Equation \eqref{eqs_f_r_0} can be rewritten as 
 \begin{align}
  f_r = \mathbf{f} \mathbf{e}_r = \mathbf{f} \mathbf{V} \mathbf{y}_r= \sum_{i=1}^{J_k-1} \mathbf{f} \mathbf{V}_i \mathbf{y}_r^i + \mathbf{f} \mathbf{V}_{J_k} \mathbf{y}_r^{J_k}.
  \label{eqs_security_f_r}
 \end{align}
 Since we assume that all of the responses from the servers $\mathcal{N}_1,\mathcal{N}_2,\dots,\mathcal{N}_{J_k-1}$ can be wiretapped, we can assume that the first part of the above equation is known while the second part is secret to the adversary. Lets define  
 $S_r \triangleq \sum_{i=1}^{J_k-1} \mathbf{f} \mathbf{V}_i \mathbf{y}_r^i$ and
 $T_r \triangleq \mathbf{f} \mathbf{V}_{J_k} \mathbf{y}_r^{J_k}. $
 The requested content can be written as $f_r = S_r + T_r$ and since all  operations are in $\mathbb{F}_2$, we have
 \begin{align}
  S_r = f_r + T_r.
  \label{eqs_security_eq1}
 \end{align}
 This is similar to the Shannon cipher system \cite{shannon1949communication} in which an encoding function $\mathfrak{e}: \mathbb{M} \times \mathbb{K} \to \mathbb{C}$ is mapping a message $\mathfrak{M} \in \mathbb{M}$ and a key $\mathfrak{K} \in \mathbb{K}$ to a codeword $\mathfrak{C} \in  \mathbb{C}$. In our problem $f_r$, $T_r$, and $S_r$ can be regarded as the message, key, and codeword respectively. 
 The eavesdropper knows the encoded file $S_r$ but it cannot obtain any information about the message $f_r$ if a unique key $T_r$ with uniform distribution is used for each message. 
 
 The following theorem provides the necessary and sufficient condition \cite{bloch2011physical} to obtain perfect  secrecy.
\begin{thm}{\em 
 If $|\mathbb{M}|=|\mathbb{K}|=|\mathbb{C}|$, a coding scheme achieves perfect secrecy if and only if 
 \begin{itemize}
  \item For each pair $(\mathfrak{M}, \mathfrak{C}) 
  \in (\mathbb{M} \times \mathbb{C})$, there exists a unique key 
  $\mathfrak{K} \in \mathbb{K}$ such that $\mathfrak{C} = 
  \mathfrak{e}(\mathfrak{M},\mathfrak{K})$.
  \item The key $\mathfrak{K}$ is uniformly distributed in $\mathbb{K}$.
 \end{itemize}
 }\label{thm_shannon}
\end{thm}
\begin{proof}
 The proof can be found in section 3.1 of \cite{bloch2011physical}.  
\end{proof}
We will use Theorem \ref{thm_shannon} to prove that our approach can achieve asymptotic perfect secrecy. 
To use this theorem, first we prove that for large enough values of $m$, the key $T_r$ is uniformly distributed.
\begin{lem}{\em
The asymptotic distribution of bits of coded files on the servers tend to uniform.
 }\label{lem_uniform_key}
\end{lem}
\begin{proof}
The proof is skipped due to page limitations. A similar proof appears in \cite{kiskani2017asecure}.
\end{proof}
This lemma paves the way to prove the following theorem. 
\begin{thm}{\em 
For the proposed full rank encoding scheme if $m$ is large but $m <2^h$, then the proposed encoded strategy provides asymptotic perfect secrecy against any eavesdropper which is capable of wiretapping all but one of the links from the servers to a user in a RG.
}\label{thm_secrecy}
\end{thm}
\begin{proof}
We formulated this problem as a Shannon cipher system assuming that $\mathfrak{M}=f_r$, $\mathfrak{K}=T_r$, and $\mathfrak{C} = S_r$. The condition $m < 2^h$ ensures that a unique vector $\mathbf{y}_r^{J_k}$ exists for each requested message. Therefore, since full rank encoding scheme is used, then $\mathbf{V}_{J_k}$ will be full rank and $T_r$ guarantees that a unique key exists for each requested message $f_r$. Notice that if the size of the RG is large enough, then the unique choice of the key does not affect the solvability of the linear equation $\mathbf{V} \mathbf{y}_r = \mathbf{e}_r$. Therefore, for any pair $(\mathfrak{m}, \mathfrak{C}) \in (\mathbb{M}, \mathbb{C})$, a unique key $\mathfrak{K} \in \mathbb{K}$ exists such that $\mathfrak{C} = \mathfrak{m} + \mathfrak{K}$. Further, we are guaranteed to have $|\mathbb{M}|=|\mathbb{K}|=|\mathbb{C}|$.
 
Notice that the key  $\mathfrak{K}=T_r$ belongs to the set of all possible bit strings with $M$ bits. Lemma \ref{lem_uniform_key} proves that each encoded file is uniformly distributed among all $M$-bit strings. Hence each key which is a unique summation of such encoded files is uniformly distributed among the set of all $M$-bit strings. In other words, regardless of the distribution of the bits in files, $T_r$ can be any bit string with equal probability for large values of $m$. Therefore, the conditions in Theorem \ref{thm_shannon} are met and perfect secrecy is achieved. 
\end{proof}
\begin{rem}{\em
 In this paper, we have assumed that the decoding vector $\mathbf{y}_r$ and the encoding matrix $\mathbf{V}_i$ are computed during the data preloading phase securely. Therefore, the eavesdropper cannot decode this information on any of the servers or have any knowledge about the key $T_r$.
 }\label{rem_key_exchange}
\end{rem}
\begin{rem}{\em 
A naive approach to achieve perfect secrecy using the Shannon cipher system is to choose $m$ different keys from the set of uniform $M$-bit strings and store them and use them to encode the files. However, since the file size $M$ is very large, this requires a significant amount of storage space to store the keys on the trusted servers which doubles the required storage capacity. The important contribution of our approach is that users do not need to store the keys and yet perfect secrecy can still be achieved with the help of trusted servers.  
}\label{rem_double_storage}
\end{rem}
 
\section{Private Information Retrieval}  
 \label{sec_pir}
 In PIR, the goal is to provide conditions that when a user downloads the content $f_r$ with index $r \in \{1,2,\dots,m\}$, the content index remains a secret to all of the servers. This is desirable in applications like Peer-to-Peer networks and in situations where some servers may have been compromised by the adversary. To achieve PIR, users send queries to the servers and servers respond to users based on those queries. These queries should be designed in a way that reveal no information to the servers about the requested content index. To formally define the {\em information theoretic PIR}, let $R$ be a random variable denoting the requested content index and let $\mathcal{Q}_l$ be a subset of at most $l$ queries. We have the following definition. 
 
 \begin{mydef}{\em 
  A PIR scheme is capable of achieving perfect information theoretic PIR against $i$ colluding servers if for the set $\mathcal{Q}_l$ of all queries available to all of these servers and any number of contents we have  
  \begin{align}
   I(R;\mathcal{Q}_l) = 0
   \label{eqs_pir}
  \end{align}
  where $I(.)$ is the mutual information function.
  }\label{def_pir}
 \end{mydef}
 
 
\subsection{Random Query Generation}
 To achieve PIR, the user chooses a fixed $\epsilon > 0$ and sets $A^{\epsilon} \triangleq m+ \lceil \log_2(\frac{1}{\epsilon}) \rceil$. Then it picks $A^{\epsilon}$ query vectors from $\mathbb{F}_2^m$ uniformly at random and statistically independent of each other. These will be the set of random queries. Therefore, we will have a set $\mathcal{Q}^{\epsilon} = \{\mathbf{q}_1, \mathbf{q}_2, \dots, \mathbf{q}_{A^{\epsilon}}\}$ of i.i.d. random query vectors. In the following, we will prove that with a probability of at least $1-\epsilon$, these random vectors span the whole $m$-dimensional space of $\mathbb{F}_2^m$. The properties of random vectors that we have used for our coding technique, had been previously studied in \cite{kolchin1999random}. 
 \begin{thm}{\em
 Let $\mathbf{Q}$ be a matrix of size $m \times l$ whose elements are independent random variables taking the values 0 and 1 with equal probability and let $\rho_m(l)$ be the rank of the matrix $\mathbf{Q}$ in $\mathbb{F}_2$. Let $s \ge 0$ and $c$ be fixed integers, $c+s \ge 0$. If $m \to \infty$ and $l=m+c$, then 
 \begin{align}
  \mathbb{P}[\rho_m(l) &= m - s] \to \nonumber \\
  & 2^{-s(s+c)} \prod_{i=s+1}^{\infty} \left(1-\frac{1}{2^i} \right) \prod_{j=1}^{s+c} \left(1-\frac{1}{2^j} \right)^{-1}
 \end{align}
 where the last product equals 1 for $c+s=0$.
 }\label{thm_kolchin}
\end{thm}
\begin{proof}
 This is Theorem 3.2.1 in page 126 of \cite{kolchin1999random}.
\end{proof}
 \begin{corol}{\em
 For $l = m + c$ where $c \ge 0$, if $m \to \infty$ we have 
 \begin{align}
 \mathbb{P}[\rho(l) = m] \to  \prod_{i=c+1}^{\infty} \left(1 - \frac{1}{2^i} \right) 
 \label{eqs_fullrank_corol}
\end{align}
 }\label{corol_kolchin}
 \end{corol}
\begin{proof}
 The proof follows for $s=0$ in Theorem \ref{thm_kolchin}.
\end{proof}
In the following, we will use these results for our proofs. 
  \begin{mydef}{\em 
  We define the random variable $A$ as the minimum number of random query vectors $\mathbf{q}_1, \mathbf{q}_2, \dots, \mathbf{q}_i$ to span the whole space of $\mathbb{F}_2^m$.
  }\label{def_2}
 \end{mydef}
 
 \begin{lem}{\em
  The probability of the event that $A < m$ is zero and for any $c \ge 0$ we have 
  \begin{align}
 \mathbb{P}[A \le m+c] \to \prod_{i=c+1}^{\infty} \left(1 - \frac{1}{2^i} \right) 
 \label{eqs_prob_a}
\end{align}
  }\label{lem_kolchin}
 \end{lem}
 \begin{proof}
  This is a direct result of Corollary \ref{corol_kolchin}.
 \end{proof} 

 \begin{lem}{\em
 The probability of the event that $A=m+c$ is less than $2^{-c}$ for any $c \ge 0$.
 }\label{lem_prob_specific1}
\end{lem}
\begin{proof}
  Let $F(c) \triangleq \mathbb{P}[A \le m+c]$. It is easy to verify from equation \eqref{eqs_prob_a} that for $m \to \infty$ we have 
 \begin{align}
 F(c) \to \frac{F(c-1)}{1-\frac{1}{2^c}}. 
 \label{eqs_fc_to_rc_1}
 \end{align}
 Since $F(c) \le 1$, from equation \eqref{eqs_fc_to_rc_1} we arrive at 
 \begin{align}
  F(c-1) \le 1 - 2^{-c}.
  \label{eqs_up_fc}
 \end{align}
 Hence,
 \begin{align}
  &\mathbb{P}[A = m + c] = F(c)-F(c-1) \to F(c-1) \left(\frac{1}{1-\frac{1}{2^c} }-1 \right) \nonumber \\ 
  &=F(c-1) \left(\frac{1}{1-\frac{1}{2^c} }-1   \right) = \frac{F(c-1)}{2^c-1} \le \frac{1-2^{-c}}{2^c-1}  = 2^{-c} \nonumber 
 \end{align}
\end{proof}
\begin{lem}{\em
 The probability of the event that $A \le m+c$ is at least $1-2^{-c}$ and at most $1-2^{-(c+1)}$ for any $c \ge 0$. i.e. 
 \begin{align}
  1-2^{-c} \le F(c) \le 1-2^{-(c+1)}
 \end{align}
 }\label{lem_prob_sum}
\end{lem}
\begin{proof}
The upper bound is already proved in equation \eqref{eqs_up_fc}. From Lemma \ref{lem_prob_specific1} we have, 
 \begin{align}
  F(c) &= \mathbb{P}[A \le m+c] = 1 - \mathbb{P}[A > m+c]  \nonumber \\
  &= 1 - \sum_{i=c+1}^{\infty} \mathbb{P}[m+i]\ge 1 - \sum_{i=c+1}^{\infty}2^{-i} = 1 - 2^{-c} \nonumber 
 \end{align}
\end{proof}  
\begin{thm}{\em
 With a probability of at least $1-\epsilon$, the set of random queries $\mathcal{Q}^{\epsilon} = \{\mathbf{q}_1, \mathbf{q}_2, \dots, \mathbf{q}_{A^{\epsilon}}\}$  where $A^{\epsilon} = m+ \lceil \log_2(\frac{1}{\epsilon}) \rceil$ spans the whole $m$-dimensional space of $\mathbb{F}_2^m$.
 }\label{thm_min_prob}
\end{thm}
\begin{proof}
 From Lemma \ref{lem_prob_sum}, we have 
 \begin{align}
  \mathbb{P}[A \le A^{\epsilon} = m+ \lceil \log_2(\frac{1}{\epsilon}) \rceil ] \ge 1-2^{-\lceil \log_2(\frac{1}{\epsilon}) \rceil} \ge 1-\epsilon \nonumber
 \end{align}
This proves the theorem.   
\end{proof}
Theorem \ref{thm_min_prob} states that the probability of spanning the $m$-dimensional space can arbitrarily go to 1 provided that the number of random vectors increases logarithmically with $\frac{1}{\epsilon}$. For example, to span the $m$-dimensional space with a probability of at least $0.99$, it is enough to only have $m+7$ random vectors. Using these random query vectors, we can now show that even with a large number of colluding servers no information about the requested content index can be obtained. To prove this result, we need to prove some lemmas. 
  

  Let $\mathbf{Q}^{\epsilon} \triangleq [\mathbf{q}_1~ \mathbf{q}_2 ~\dots ~\mathbf{q}_{A^{\epsilon}}]$ be the matrix of size $m \times A^{\epsilon}$ whose columns are random query vectors. Matrix $\mathbf{Q}^{\epsilon}$ contains $A^{\epsilon}$ statistically independent random vectors. Let $B_{\mathbf{x}}^r$ be the event that for a specific vector $\mathbf{x} \in \mathbb{F}_2^{A^{\epsilon}}$ and a specific base vector $\mathbf{e}_r$, we have $\mathbf{Q}^{{\epsilon}} \mathbf{x} = \mathbf{e}_r$.
 \begin{lem}{\em 
 For any specific non-zero vector $\mathbf{x} \in \mathbb{F}_2^{A^{\epsilon}}$ we have 
 \begin{align}
  \mathbb{P}[B_{\mathbf{x}}^r] = \mathbb{P}[\mathbf{Q}^{\epsilon} \mathbf{x} = \mathbf{e}_r] = 2^{-m} \label{eqs_lem_fix_x}.
 \end{align}
  }\label{lem_fix_x}
 \end{lem}
\begin{proof} 
Lets assume vector $\mathbf{x}$ has $k$ ones. 
  If $\mathbf{Q}^{\epsilon} \mathbf{x} = \mathbf{e}_r$, then $k$ vectors from the set of all vectors $\mathbf{q}_1, \mathbf{q}_2, \dots, \mathbf{q}_{A^{\epsilon}}$ are added together to create $\mathbf{e}_r$. Lets denote these vectors by $\mathbf{q}_{e_1}, \mathbf{q}_{e_2}, \dots, \mathbf{q}_{e_k}$. Let $q^{e_j}_{r}$ denote the $r^{th}$ element of vector $\mathbf{q}_{e_j}$.
 Since the vectors $\mathbf{q}_{e_1}, \mathbf{q}_{e_2}, \dots, \mathbf{q}_{e_k}$ are independent and their elements are also mutually independent, using binary summations in $\mathbb{F}_2$, we have 
 \begin{align}
  \mathbb{P}[B_{\mathbf{x}}^r] = \mathbb{P}[\mathbf{Q}^{\epsilon} \mathbf{x} = \mathbf{e}_r ] & = 
  \mathbb{P}[\sum_{j=1}^k q^{e_j}_{r}= 1] \prod_{\substack{l'=1\\l' \neq r}}^m
  \mathbb{P}[\sum_{j=1}^k q^{e_j}_{l'} = 0].
  \label{eqs_lem_fix_x_proof} 
 \end{align}
 We can easily prove that $\mathbb{P}[\sum_{j=1}^k q^{e_j}_{r} = 1] = \frac{1}{2}$.
To prove this, we can use induction on $k$. This equation is valid for the base case $k=1$. Assume that it is valid for $k-1$. We have 
  \begin{align}
  \mathbb{P}[\sum_{j=1}^k q^{e_j}_{r}= 1] &=  
  \mathbb{P}[q^{e_k}_{r} = 1]\mathbb{P}[\sum_{j=1}^{k-1} q^{e_j}_{r} = 0] \nonumber \\
  &+   \mathbb{P}[q^{e_k}_{r} = 0]\mathbb{P}[\sum_{j=1}^{k-1} q^{e_j}_{r} = 1] 
  = \frac{1}{2}  
  \nonumber 
 \end{align}
 Similarly, it is easy to prove that $ \mathbb{P}[\sum_{j=1}^k q^{e_j}_{l'} = 0] = \frac{1}{2}$.
 Hence, equation \eqref{eqs_lem_fix_x_proof} can be simplified to $ \mathbb{P}[B_{\mathbf{x}}^r] = \mathbb{P}[\mathbf{Q}^{\epsilon} \mathbf{x} = \mathbf{e}_r] = 2^{-m}$.
 \end{proof}
\begin{lem}{\em 
  The following inequalities hold for $1 \le j \le i$, 
  \begin{align}
  \frac{1}{i+1} 2^{i H(\frac{j}{i})} \le  \binom{i}{j} \le 2^{i H(\frac{j}{i})}
  \end{align}
 where $H(\alpha)$ denotes the binary entropy function, i.e. $H(\alpha) = -\alpha \log_2(\alpha) - (1 - \alpha) \log_2(1-\alpha)$.
}\label{lem_binom_entropy} 
\end{lem}
\begin{proof}
The proof can be found in the appendix of \cite{DBLP:journals/tit/MacKay99}.
\end{proof}   
We are now ready to prove the following theorem which shows that accessing a significant number of random queries in $\mathcal{Q}^{\epsilon}$ cannot help in reconstructing any of the base vectors for large $m$.
 \begin{thm}{\em  
  Consider the set $\mathcal{Q}^{\epsilon} = \{\mathbf{q}_1, \mathbf{q}_2,\dots,\mathbf{q}_{A^{\epsilon}}\}$ of $A^{\epsilon} = m + \lceil \log_2(\frac{1}{\epsilon}) \rceil$ statistically independent random uniform query vectors. For large enough values of $m$ with probability arbitrarily close to 1, none of the base vectors exist in the span of any subset $\mathcal{Q}_l \subset \mathcal{Q}^{\epsilon}$ with cardinality of at most $l = \lfloor \delta m \rfloor$ where $\delta < 0.5$.
 }\label{thm_spy}
 \end{thm}
 \begin{proof}
  Consider any base vector $\mathbf{e}_r$ and a non-zero vector $\mathbf{x} \in \mathbb{F}_2^{A^{\epsilon}}$. For this vector, computing $\mathbf{Q}^{\epsilon} \mathbf{x}$ in $\mathbb{F}_2$ is equivalent to adding a subset of columns of $\mathbf{Q}^{\epsilon}$ whose set of indices is equal to the set of indices of non-zero elements in $\mathbf{x}$. If $\mathbf{Q}^{\epsilon} \mathbf{x} = \mathbf{e}_r$ for some $\mathbf{x} \in \mathbb{F}_2^{A^{\epsilon}}$, then any subset $\mathcal{Q}_l \subset \mathcal{Q}^{\epsilon}$ which contains the column vectors of $\mathbf{Q}^{\epsilon}$ whose set of indices is equal to the set of indices of non-zero elements in $\mathbf{x}$ also spans $\mathbf{e}_r$. In fact, the number of non-zero elements of $\mathbf{x}$ or Hamming weight of $\mathbf{x}$ (i.e., Ham($\mathbf{x}$)) is equal to the number of vectors that should be added to reconstruct $\mathbf{e}_r$.
    
  Consider all vectors $\mathbf{x} \in \mathbb{F}_2^{A^{\epsilon}}$ with Hamming weight less than or equal to $l= \lfloor \delta m \rfloor$ where $\delta < 0.5$. Lemma \ref{lem_fix_x} shows that for any $\mathbf{x}$, we have $\mathbb{P}[B_{\mathbf{x}}^r] = 2^{-m}$. Therefore, the asymptotic probability of existence of a subset $\mathcal{Q}_l \subset \mathcal{Q}^{\epsilon}$ with a cardinality of at most $l = \lfloor \delta m \rfloor$ which spans $\mathbf{e}_r$ for large values of $m$ can be found as 
  \begin{align}
   &\lim_{m \to \infty}  \mathbb{P}[\exists \mathcal{Q}_l \subseteq \mathcal{Q}^{\epsilon} | \textrm{card}\{\mathcal{Q}_l\} \le l  = \lfloor \delta m \rfloor,\mathbf{e}_r \in \textrm{span} \{ \mathcal{Q}_l\}] ,
   \nonumber \\
   =&\lim_{m \to \infty}  \mathbb{P}[\bigcup_{\substack{\mathbf{x} \in \mathbb{F}_2^{A^{\epsilon}}, \textrm{Ham}(\mathbf{x}) \le l}} B_{\mathbf{x}}^r ]  \stackrel{(a)}{\le} \lim_{m \to \infty}  \sum_{\substack{\mathbf{x} \in \mathbb{F}_2^{A^{\epsilon}}, \textrm{Ham}(\mathbf{x}) \le l}} \mathbb{P}[B_{\mathbf{x}}^r],
   \nonumber \\
   \stackrel{(b)}{=}& \lim_{m \to \infty}  \sum_{i=1}^l \binom{A^{\epsilon}}{i} 2^{-m} \stackrel{(c)}{\le} \lim_{m \to \infty} 
    l \binom{A^{\epsilon}}{l} 2^{-m}, \nonumber \\
    \stackrel{(d)}{\le}& \lim_{m \to \infty}  l \binom{m}{l} 2^{-m} 
   = \lim_{m \to \infty} \lfloor \delta m \rfloor \binom{m}{\lfloor \delta m \rfloor} 2^{-m}, \nonumber \\
   \stackrel{(e)}{\le}& \lim_{m \to \infty}  \delta m 2^{mH\left(\frac{\lfloor \delta m \rfloor}{m} \right)} 2^{-m} \stackrel{(f)}{\le} 
   \lim_{m \to \infty}  \delta m 2^{-m (1-H(\delta))} \stackrel{(g)}{=} 0,
   \nonumber 
  \end{align}
  where inequality (a) comes from the {\em union bound} and (b) holds by using Lemma \ref{lem_fix_x} and counting all the vectors $\mathbf{x}$ with Hamming weight less than $l = \lfloor \delta m \rfloor$ and  inequality (e) comes from Lemma \ref{lem_binom_entropy}. Notice that (c), (d), (e) and (f) are only valid for cases when $\delta < 0.5$. This shows that the probability of existence of any desired base in the span of any subset of vectors with cardinality less than $\lfloor \delta m \rfloor$ goes to zero as $m$ grows if $\delta < 0.5$. 
 \end{proof}

 \begin{rem}{\em
 In practice the user generates enough number of random vectors  to span the whole $m$-dimensional space. Hence, it has a set $\mathcal{Q} = \{\mathbf{q}_1, \mathbf{q}_2, \dots, \mathbf{q}_A \}$ of $A \ge m$ total random vectors. Then it chooses a subset $\mathcal{Q}^{\textrm{full}} = \{\mathbf{q}_{t_1}, \mathbf{q}_{t_2}, \dots, \mathbf{q}_{t_m} \} \subseteq \mathcal{Q}$ of $m$ linearly independent vectors from them and use them as its query vectors. This way it is guaranteed that the $m$ queries will span the whole space of  $\mathbb{F}_2^m$ and any base vector $\mathbf{e}_r$ can be represented in terms of these independent query vectors as
 \begin{equation}
 \mathbf{e}_r = \sum_{k=1}^m d_k \mathbf{q}_{t_k}
 \label{eq_query_expansion}
 \end{equation}
 The following lemma shows that the average required number of random vectors to span $\mathbb{F}_2^m$ is very close to $m$ so in practice only a few number of random queries more than $m$ is needed to span $\mathbb{F}_2^m$.  
  }\label{rem_practical_query}
 \end{rem}
 \begin{lem}{\em
 If ${\bf q}_j$ is a random vector belonging to $\mathbb{F}_2^{{m}}$ with elements having uniform distribution, the average minimum number of vectors ${\bf q}_j$ to span the whole space of $\mathbb{F}_2^{{m}}$ equals  
 \begin{equation}
  \mathbb{E}_q = {m} + \sum_{i=1}^{m} \frac{1}{2^{i}-1}  =   {m} + \gamma,
\label{gamma}
 \end{equation}
 where $\gamma$ asymptotically approaches  the Erdős–Borwein constant ($\approx 1.6067$).
 }\label{lem_mohsen}
\end{lem}
\begin{proof}
 The proof can be found in \cite{DBLP:journals/twc/KiskaniS17}.
\end{proof}
  \begin{rem}{\em
  Since $\mathcal{Q}^{\textrm{full}} \subseteq \mathcal{Q}$, if any vector  $\mathbf{e}_r$ does not exist in the span of any subset $\mathcal{Q}_l \subset \mathcal{Q}$, of $l$ random query vectors, it will not exist in the span of any subset $\mathcal{Q}_l \subset \mathcal{Q}^{\textrm{full}}$ of $l$ random query vectors in $\mathcal{Q}^{\textrm{full}}$ too. So, Theorem \ref{thm_spy}  remains valid for this choice of random queries too. This means that in practice, every base vector is guaranteed to exist in the span of the $m$ query vectors but none of the base vectors exist in the span of any subset $\mathcal{Q}_l \subset \mathcal{Q}$ with probability close to one if $l < \lfloor \delta m \rfloor$ for $\delta < 0.5$.
 }\label{rem_k1}
 \end{rem}
  
 \subsection{Responding to Queries}
 In this section, we assume that the user has chosen $m$ linearly independent random query vectors in $\mathcal{Q}^{\textrm{full}}$ and wants to download the $r^{th}$ content. Since $\mathcal{Q}^{\textrm{full}}$ is a set of vectors which spans the whole space of $\mathbb{F}_2^m$, the user can expand the base vector $\mathbf{e}_r$ in terms of the query vectors in $\mathcal{Q}^{\textrm{full}}$ as mentioned in \eqref{eq_query_expansion}. Hence, the requested content can be expanded in terms of query vectors as 
 \begin{align}
  f_r = \mathbf{f} \mathbf{e}_r = \mathbf{f} \left(\sum_{k=1}^m d_k  \mathbf{q}_{t_k} \right) &= \sum_{k=1}^{m} d_k \mathbf{f} \mathbf{q}_{t_k} 
  \label{eqs_reconstruction}
 \end{align} 
 where $d_k \in \mathbb{F}_2$ is either 0 or 1. Based on equation \eqref{eqs_reconstruction} the user requests some parts of the desired content from each RG so that none of the RGs can understand any information about the requested content.  
 
 To accomplish PIR, the user partitions the set of random queries $\mathbf{q}_{t_k}$ whose corresponding decoding gains $d_k$ are non-zero into $a$ disjoint subsets $\mathcal{Q}_1, \mathcal{Q}_2, \dots, \mathcal{Q}_a$. The choice of number of subsets (i.e. $a$) depends on the number of colluding servers. Each subset of queries is then sent to a different RG as depicted in Figure \ref{fig_query_response}. Therefore, the requested content can be retrieved as 
  \begin{align}
  f_r = \sum_{\substack{\mathbf{q}_{t_k} \in \mathcal{Q}^{\textrm{full}} \\d_k \neq 0}}  \mathbf{f} \mathbf{q}_{t_k} = \sum_{i=1}^a \sum_{\mathbf{q}_{t_k} \in \mathcal{Q}_i }  \mathbf{f} \mathbf{q}_{t_k}
  \label{eqs_reconstruction2}
 \end{align}
 The ultimate goal in PIR is to prevent any colluding group of servers to gain information about the requested content index. Assume that the number of colluding servers is $b$. If any two colluding servers lie within the same RG, they receive the same subset of queries from the user. Therefore, without loss of generality we consider the worst scenario in which all the colluding servers lie within different RGs and all these $b$ colluding servers are able to collaboratively obtain all the queries  $\mathcal{Q}_1, \mathcal{Q}_2, \dots, \mathcal{Q}_b$.  Based on Theorem \ref{thm_spy}, if the number of all query vectors in $\mathcal{Q}_l = \cup_{i=1}^b \mathcal{Q}_i$ is less than $\lfloor \delta m \rfloor$ for some $\delta < 0.5$, then no information can be achieved about the requested content index. This provides significant PIR capability for this technique.
   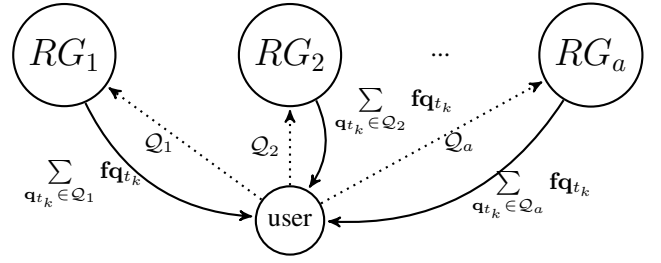
\begin{figure}
 \centering
 \begin{tikzpicture}[->,>=stealth',shorten >=1pt,auto,
                    thick,main node/.style={circle,draw,font=\sffamily\Large\bfseries}]

  \node[main node] (1) {$RG_1$};
  \node[main node] (2) [right of=1,node distance=3cm] {$RG_2$};
  \node (y) [right of=2, node distance = 2cm] {...};
  \node (x) [draw, circle, thick, below of=2, node distance=2.2cm] {user};
  \node[main node] (4) [right of=2, node distance=4cm] {$RG_a$};

  \path[every node/.style={font=\sffamily\small}]
    (1) edge [bend right] node [left] {$\sum\limits_{ \mathbf{q}_{t_k} \in \mathcal{Q}_1} \mathbf{f} \mathbf{q}_{t_k} $} (x)
    (2) edge [bend left] node [above right] {$\sum\limits_{ \mathbf{q}_{t_k} \in \mathcal{Q}_2} \mathbf{f} \mathbf{q}_{t_k} $} (x)
    (4) edge [bend left] node [right] {$\sum\limits_{ \mathbf{q}_{t_k} \in \mathcal{Q}_a} \mathbf{f} \mathbf{q}_{t_k} $} (x)
    (x) edge [dotted] node [left] {$\mathcal{Q}_1$} (1)
    (x) edge [dotted] node [left] {$\mathcal{Q}_2$} (2)
    (x) edge [dotted] node [right] {$\mathcal{Q}_a$} (4);
\end{tikzpicture}
\caption{Multiple RGs respond to queries sent from the user. This allows the
user to privately download its desired content.} 
\label{fig_query_response}
\end{figure}

 Notice that since RGs have full rank encoding matrices, they can respond to any query that they receive. Assume that RG $\mathcal{J}_i$ with the full rank encoding matrix $\mathbf{V}  = [\mathbf{V}_1 ~\mathbf{V}_2~ \dots ~\mathbf{V}_{J_i}]$ receives the set of queries $\mathcal{Q}_i$. This RG needs to send $ \sum_{ \mathbf{q}_{t_k} \in \mathcal{Q}_i} \mathbf{f} \mathbf{q}_{t_k}$ to the user. It can solve the linear equation 
 \begin{align}
  \mathbf{V} \mathbf{p}_i = \sum_{ \mathbf{q}_{t_k} \in \mathcal{Q}_i}  \mathbf{q}_{t_k}
  \label{eqs_lin_eq2}
 \end{align}
 in Galois Field $\mathbb{F}_2$ for $\mathbf{p}_i$ as 
  \begin{align}
  \mathbf{p}_i = \mathbf{V}^T \left(\mathbf{V} \mathbf{V}^T \right)^{-1} \left( \sum_{ \mathbf{q}_{t_k} \in \mathcal{Q}_i} \mathbf{q}_{t_k} \right).
  \label{eqs_lin_eq2_soln}
 \end{align}
 Similar to before the server $\mathcal{N}_s$ in the RG $\mathcal{J}_i$ which has already acquired all the information in matrix $\mathbf{V}$, computes the overal query decoding solution $\mathbf{p}_i$ which is a vector of size $J_i h \times 1$. If this vector is divided into $J_i$ equal size pieces as $\mathbf{p}_i = [\mathbf{p}_i^1~\mathbf{p}_i^2~\dots~\mathbf{p}_i^{J_i}]^T$, then the server $\mathcal{N}_s$ sends the $j^{th}$ portion of $\mathbf{p}_i$ to server $\mathcal{N}_j$ in the RG $\mathcal{J}_i$. More precisely, server $\mathcal{N}_j$ receives a query response vector $\mathbf{p}_i^j$ of size $h \times 1$ from $\mathcal{N}_s$ for each $j=1,2,\dots,J_i$. Then the server $\mathcal{N}_j$ sends $\mathbf{f} \mathbf{V}_j \mathbf{p}_i^j$ back to the coordinating server $\mathcal{N}_s$. The coordinating server $\mathcal{N}_s$ then aggregates all the data received from multiple servers in the RG to construct $\sum_{ \mathbf{q}_{t_k} \in \mathcal{Q}_i} \mathbf{f} \mathbf{q}_{t_k}$ as
 \begin{align}
  \sum_{ \mathbf{q}_{t_k} \in \mathcal{Q}_i} \mathbf{f} \mathbf{q}_{t_k} = \mathbf{f} \mathbf{V} \mathbf{p}_i = \sum_{j=1}^{J_i} \mathbf{f} \mathbf{V}_j \mathbf{p}_i^j.
  \label{eqs_reconstruction_group}
 \end{align}
 The coordinating server $\mathcal{N}_s$ in the RG $\mathcal{J}_i$ then transmits
 $\sum_{ \mathbf{q}_{t_k} \in \mathcal{Q}_i} \mathbf{f} \mathbf{q}_{t_k}$ to the user. 
 
 Each RG only transmits one encoded file to the user. However all the servers within an RG need to collaborate with each other prior to responding to the queries sent from the user. Notice that  communication between the servers are carried using high bandwidth fiber optic links while transmissions from the servers to the user  are performed over low bandwidth links. In our computation of communication cost for achieving PIR, we only consider 
 communication between the servers and the user in the low bandwidth links.   

 \begin{rem}{\em
It is worth mentioning that in this approach the coordinations between  servers in an RG is necessary because  the servers do not have full storage capacity to store all the contents. In fact if we also assume each server has high storage capacity similar to \cite{tajeddine2016private}, then each server can act as an RG and there will be no communications between servers.
  }\label{rem_full_size}
 \end{rem}

\subsection{Trade-off Between Communication Cost and Privacy Level}

In order to achieve PIR, each user needs to download more information. This additional bandwidth utilization is referred to as {\em communication Price of Privacy (cPoP)} \cite{tajeddine2016private} which is defined as follows. Note that the cost of sending queries are ignored because it is assumed that the size of contents are significantly higher than the size of the queries.
%
\begin{mydef}{\em
 The communication Price of Privacy (cPoP) is the ratio of the total number of bits downloaded by the user from the servers to the size of the requested file. 
 }\label{def_cpop}
\end{mydef}
 To explain the trade-off between communication cost and level of privacy, assume that the user divides the queries into $a$ equal size groups of queries and sends each group of queries to a different RG. Each RG should respond to at most $\lceil \frac{m}{a} \rceil$ queries. If $b$ RGs collude to gain some information about the requested content index, then they will have access to a total of at most $b \lceil \frac{m}{a} \rceil$ queries. We proved that knowing $\lfloor \delta m \rfloor$ queries asymptotically gives no information about the requested content index if $\delta < 0.5$. Hence, if $b < \frac{a}{2}$, then the colluding RGs will get no information about the requested content index. Therefore, if less than half of the RGs collude to gain some information about the requested content index, they cannot gain any information. We can increase $a$ to get the maximum possible level of privacy. However, the downside of increasing $a$ is that the communication Price of Privacy (cPoP) will also increase.

 As discussed earlier, if the queries are sent to $a$ RGs then $a$ responses from these RGs are required to retrieve a content. Since each RG transmits an encoded file of size $M$ bits to the user the total number of bits downloaded by the user will be equal to $aM$ and therefore the cPoP will be equal to $aM/M = a$. 

\subsection{Full Size Servers}
 
 Assume that the servers have large storage capability such that each RG is only composed of 1 server. Our assumption of full rank encoding scheme guarantees that servers with storage ability of $h \geq m$ encoded files can be used to retrieve any desired content. In \cite{tajeddine2016private}, the authors studied the use of MDS codes for PIR. They considered full size storage systems with MDS codes and they considered the case when only one of the databases is compromised.  They proposed a PIR technique in which a cPoP of $\frac{1}{1-R}$ can be achieved in full size databases where $R$ is the MDS code rate. To compare our results with \cite{tajeddine2016private}, notice that if we assume that there is only one malicious server in the cloud, then we can choose any two servers and send half of the queries to each one of them. This way we have a cPoP of 2 which is better than the results in \cite{tajeddine2016private} for $R > 1/2$. 
 
 \section{Simulation}
 \label{sec_sim}
 
 To numerically verify the results proved in section \ref{sec_pir}, we created $m$ linearly independent random query vectors which are used to expand the bases. Figure \ref{fig_queries} demonstrates the probability of the event that at least one of the base vectors exists in the span of $l = \lfloor \delta m \rfloor$ vectors for $\delta = 0.1,0.2,0.3$ and $0.4$. Consistent with our results in section \ref{sec_pir}, the probability of the event that a base exists in the span of any set of 
 $l = \lfloor \delta m \rfloor$ vectors goes quickly to zero. 
 
 It is proved  \cite{DBLP:journals/tit/Vardy97} that the problem of finding the minimum spanning set of vectors is NP-Complete. It is even proved \cite{DBLP:journals/tit/DumerMS03} that this problem is NP-Hard to approximate. Therefore, in general it is NP-Hard to find out if a given base exists in the span of at most $l = \lfloor \delta m \rfloor$ vectors out of the $m$ vectors. For our simulations we have used a brute force approach to check if a given base exists in the span of at most $l = \lfloor \delta m \rfloor$ vectors out of the $m$ existing random query vectors where $m \le 20$.   
 \begin{figure}
    \center
      \includegraphics[scale=0.4,angle=0]{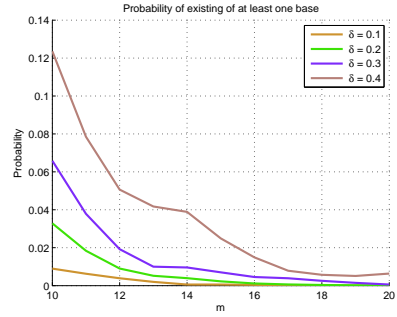}\\
      \caption{Probability of the event that at least one base exists in the span of any subset of $l=\lfloor \delta m \rfloor$ random vectors.}
    \label{fig_queries}
\end{figure}
 
 
\section{Conclusions}
\label{sec_conc}
In this paper, we have studied the problems of security and private information retrieval in distributed storage systems which are using a full rank encoding scheme based on Random Linear Fountain (RLF) codes. We have proposed an approach based on uniform random queries to achieve information theoretic PIR property. We have proved that our proposed technique can asymptotically achieve perfect secrecy for a distributed storage system. Our proposed solution is robust against a significant number of colluding servers in the network. We have also shown that our technique can outperform MDS codes for storage systems in terms of PIR cost for certain regimes. 
 
\bibliographystyle{unsrt}
\bibliography{PIR-Bib} 
\end{document}